\documentclass{llncs}
\usepackage{proof}
\usepackage{wrapfig}

\long\def\ignore#1{}

\newcommand{\ra}{\rightarrow}
\newcommand{\lra}{\longrightarrow}

\newcommand{\lolli}{\mbox{$-\!\circ$}}

\newcommand{\vdashi}{\vdash_{\!I}}

\newcommand{\qqob}{\quest Q^{0\!^{\bot}}}
\newcommand{\qob}[1]{\quest (#1^0)^\bot}
\newcommand{\fgob}{\vdash (\Gamma^0)^\bot}

\newcommand{\uob}[1]{(#1^0)^\bot}

\newcommand{\none}{{-\!1}}
\newcommand{\pone}{{+1}}

\newcommand{\fgnb}{\vdash (\Gamma^\none)^\bot}

\newcommand{\unb}[1]{(#1^\none)^\bot}

\newcommand{\leftfcs}[1]{\stackrel{#1}{\longrightarrow}}
\newcommand{\rightfcs}[1]{\mathrel{-_{\kern .6pt #1\/}\kern -.1pt\mathord{\rightarrow}}}

\newcommand{\wedgep}{\wedge^{\!+}}
\newcommand{\wedgen}{\wedge^{\!-}}
\newcommand{\veep}{\vee^{\!+}}
\newcommand{\veen}{\vee^{\!-}}
\newcommand{\implp}{\supset^{\!+}}
\newcommand{\impln}{\supset^{\!-}}
\newcommand{\boxr}{\mbox{$[]_r$}}
\newcommand{\boxl}{\mbox{$[]_l$}}

\newcommand{\ljf}{\hbox{\sl LJF}}

\newcommand{\lkf}{\hbox{\sl LKF}}
\newcommand{\ddp}[1]{#1^{\approx}}
\newcommand{\negg}{\sim\!}

\newcommand{\mkpos}[1]{\partial^+(#1)}
\newcommand{\mkneg}[1]{\partial^-(#1)}

\newcount\PLv\newcount\PLw\newcount\PLx\newcount\PLy\newdimen\PLyy\newdimen\PLX
\newbox\PLdot \setbox\PLdot\hbox{\tiny.} \def\scl{.08} 
\def\PLot#1{\PLx`#1\advance\PLx-42\PLy\PLx\PLv\PLx\divide\PLy9\PLw\PLy\multiply
\PLw9\advance\PLx-\PLw\advance\PLx-4\PLy-\PLy\advance\PLy4\PLX=\the\PLx pt
\advance\PLyy\the\PLy pt\wd\PLdot=\scl\PLX\raise\scl\PLyy\copy\PLdot}
\def\draw#1{\ifx#1\end\let\next=\relax\else\PLot#1\let\next=\draw\fi\next}

\def\invamp{\hbox{\PLyy=70pt\draw :::;DMV_gqppyyyyyooooxxxnnwvlutkjaWNE=5-./9
9:::CCCC:::99/..--544=EENWWaajjjkktttttttNNNVVVVVVVV\end \hskip8pt}}
\newbox\iabox\setbox\iabox\invamp \def\Invamp{\copy\iabox}

\def\relbar{\mathrel{\smash-}}
\def\joinrelm{\mathrel{\mkern-3mu}}
\def\tailpiece{\kern 1pt\vrule height 1ex width 0.3ex depth -.3ex}
\def\seqsym{\mathrel{\tailpiece\joinrelm\relbar}}

\newcommand{\true}{\hbox{\sl true}}
\newcommand{\false}{\hbox{\sl false}}
\newcommand{\bang}{\mathop{!}}
\newcommand{\quest}{\mathop{?}}
\newcommand{\with}{\mathbin{\&}}
\newcommand{\ot}{\otimes}

\newcount\PLv\newcount\PLw\newcount\PLx\newcount\PLy\newdimen\PLyy\newdimen\PLX
\newbox\PLdot \setbox\PLdot\hbox{\tiny.} \def\scl{.08} 
\def\PLot#1{\PLx`#1\advance\PLx-42\PLy\PLx\PLv\PLx\divide\PLy9\PLw\PLy\multiply
\PLw9\advance\PLx-\PLw\advance\PLx-4\PLy-\PLy\advance\PLy4\PLX=\the\PLx pt
\advance\PLyy\the\PLy pt\wd\PLdot=\scl\PLX\raise\scl\PLyy\copy\PLdot}
\def\draw#1{\ifx#1\end\let\next=\relax\else\PLot#1\let\next=\draw\fi\next}

\def\invamp{\hbox{\PLyy=70pt\draw :::;DMV_gqppyyyyyooooxxxnnwvlutkjaWNE=5-./9
9:::CCCC:::99/..--544=EENWWaajjjkktttttttNNNVVVVVVVV\end \hskip4pt}}
\newbox\iabox\setbox\iabox\invamp \def\Invamp{\copy\iabox}

\newcommand\lpar{\mathrel{\Invamp\;}}
\newcommand{\limp}{\mathbin{-\hspace{-0.70mm}\circ}}

\newcommand{\one}{{\bf 1}}
\newcommand{\zero}{{\bf 0}}
\newcommand{\nng}[1]    {#1^\perp}

\newcommand{\tseqsym}{\mathrel{\kern3pt\seqsym\kern-10pt\seqsym}}

\newcommand{\oneseq}[1]{\seqsym #1}
\newcommand{\twoseq}[2]{\strut#1 \seqsym #2}
\newcommand{\Up}[3]{#1\colon#2\Uparrow #3}
\newcommand{\Down}[3]{#1\colon#2\Downarrow #3}

\newcommand{\myhead}[1]{\noalign{\hbox to \the\hsize{\hfill\bf #1\strut\qquad\hfill}}}
\newcommand{\fib}[1]{\hbox{\sl fib}(#1)}
\newcommand{\LJQp}{\hbox{\sl LJQ$'$}}
\newbox\tallstrutbox 
\setbox\tallstrutbox=\hbox{\vrule height10.5pt depth3.5pt width0pt}

\begin{document}
\pagestyle{plain}
\title{Focusing and Polarization in Intuitionistic Logic}
\author{Chuck Liang\inst{1} \and Dale Miller\inst{2}}
\institute{Department of Computer Science, 
Hofstra University, Hempstead, NY 11550\\
\email{chuck.liang at hofstra.edu}\\ 
\and
INRIA \& LIX/Ecole Polytechnique, Palaiseau, France\\
\email{dale.miller at inria.fr}
}

\maketitle

\begin{abstract}
A focused proof system provides a normal form to cut-free proofs that
structures the application of invertible and non-invertible inference
rules.
The focused proof system of Andreoli for linear logic has been applied
to both the {\em proof search} and the {\em proof
normalization} approaches to computation.  Various proof systems 
in literature exhibit characteristics of focusing to one degree or another.
We present a new, focused proof system for intuitionistic logic, called
\ljf, and show how other proof systems can be mapped
into the new system by inserting logical connectives that prematurely
stop focusing.  We also use \ljf\ to design a focused proof system for
classical logic.  Our approach to the design and analysis of these
systems is based on the completeness of focusing in
linear logic and on the notion of polarity that appears in Girard's LC
and LU proof systems.
\end{abstract}

\section{Introduction}

Cut-elimination provides an important normal form for
sequent calculus proofs.  But what normal forms can we uncover about
the structure of cut-free proofs?  Since cut-free proofs play
important roles in the foundations of computation, such normal forms
might find a range of applications in the proof normalization
foundations for functional programming or in the proof search
foundations of logic programming.

\subsection{About focusing}
\label{about}

Andreoli's {\em focusing} proof system for linear logic (the {\em
triadic} proof system of \cite{andreoli92jlc}) provides a
normal form for cut-free proofs in linear logic.  Although we describe
this system, here called {\em LLF,\/} in more detail in
Section~\ref{llfocus}, we highlight two aspect of focusing proofs
here.  First, linear logic connectives can be divided into the {\em
asynchronous} connectives, whose right-introduction rules are
invertible, and the {\em synchronous} connectives, whose right
introduction rules are not (generally) invertible.  The search for a
focused proof can capitalize on this classification by applying
(reading inference rules from conclusion to premise) all
invertible rules in any order (without the need for backtracking) and
by applying a chain of non-invertible rules that focus on a given formula and
its positive subformulas.  Such a chain of applications, usually
called a {\em focus}, terminates when it reaches an asynchronous
formula.  Proof search can then alternate between applications of
asynchronous introduction rules and chains of synchronous introduction
rules.

A second aspect of focusing proofs is that the
synchronous/asynchronous classification of non-atomic formulas must be
extended to atomic formulas.  The arbitrary assignment of positive
(synchronous) and negative (asynchronous) {\em bias} to atomic
formulas can have a major impact on, not the existence of focused
proofs, but the shape of focused proofs.  For example, consider the
Horn clause specification of the Fibonacci series:
$$ \fib{0,0}\land \fib{1,1}\land
   \forall n\forall f\forall f'
   [\fib{n,f}\land\fib{n+1,f'}\supset \fib{n+2,f+f'}].
$$
If all atomic formulas are given negative bias, then the only focused
proofs of $\fib{n,f_n}$ are those that can be classified as ``backward
chaining'' (the size of the smallest one being exponential in $n$).
On the other hand, if all atomic formulas are given positive bias,
then the only focused proofs are those that can be classified as
``forward chaining'' (the size of the smallest one being linear in
$n$).

\subsection{Results}

The contributions of this paper are the following.  First, we
introduce in Section~\ref{ljf} a new focusing proof system \ljf\ and
show that it is sound and complete for intuitionistic logic.  Notable
features of \ljf\ are that it allows for atoms of different bias and
it contains two versions of conjunction: while these conjunctions are
logically equivalent, they are affected by focusing differently.
Second, in Section~\ref{framework}, we show how several other focusing
proof systems can be captured in \ljf, in the sense of {\em full
completeness} (one-to-one correspondence between proofs in different
systems).  One should note that while there are many focusing proof
systems for intuitionistic logic in the literature, we appear to be
the first to provide a single (intuitionistic) framework for capturing
many of them.  Third, in Section~\ref{classical}, we use \ljf\ to
derive \lkf, a focusing system for classical logic.

\subsection{Methodology and Related work}
\label{related}

There are a number of sequent calculus proof systems known to be
complete for intuitionistic logic that exhibit characteristics of
focusing.
Some of these proof systems are based on fixing globally on either
forward chaining or backward chaining.  The early work on {\it uniform
proofs} \cite{miller91apal} and the {\em LJT\/} proof system
\cite{herbelin95phd} are both backward chaining calculi (all atoms
have negative bias).  The {\em LJQ\/} calculus
\cite{herbelin95phd,dyckhoff06cie} similarly selects the global
preference to be forward chaining (all atoms have positive bias).
Less has been published about systems that allow for mixing bias on
atoms.  The $\lambda$RCC proof system of Jagadeesan, Nadathur, and
Saraswat \cite{jagadeesan05fsttcs} allows for both forward chaining
and backward chaining in a superset of the hereditary Harrop fragment
of intuitionistic logic.  Chaudhuri, Pfenning, and Price in
\cite{chaudhuri06ijcar} observed that focusing proofs with mixed
biases on atoms can form a declarative basis for mixing forward and
backward chaining within Horn clauses.  The PhD theses of Howe
\cite{howe98phd} and Chaudhuri \cite{chaudhuri06phd} also explored
various focusing proof systems for linear and intuitionistic logic.

We are interested in providing a flexible and unifying framework that
can collect together important aspects of many of these proof systems.
There are several ways to motivate and validate the design of such a
system.  One approach stays entirely within intuitionistic logic and
works directly with invertibility and permutability of inference
rules.  Such an approach has been taken in many papers, such as
\cite{miller91apal,pfenning04ln,dyckhoff06cie}.  Our approach uses
linear logic, with its exponential operators $\bang$ and $\quest$, as
a unifying framework for looking at intuitionistic (and classical)
logic.  The fact that Andreoli's focused system was defined for full
linear logic provides us with a convenient platform for exploring the
issues around focusing and polarity.  We translate intuitionistic
logic into linear logic, then show that proof systems for
intuitionistic logic match focused proofs of the translated image
(Section~\ref{trans}).  A crucial aspect of understanding focusing in
intuitionistic logic is provided by identifying the precise
relationship between Andreoli's notion of polarity with Girard's
notion of polarity found in the LC \cite{girard91mscs} and LU
\cite{girard93apal} systems (Section~\ref{polarity}).

Another system concerning polarity and focusing is found in the work
of Danos, Joinet and Schellinx \cite{danos93wll,danos97jsl}. Many
techniques that they developed, such as {\em inductive decorations,\/}
are used throughout our analysis.  Our work diverges from theirs in
the adaptation of Andreoli's system (LLF) as our main instrument of
construction.  The $LK^\eta_p$ system of \cite{danos97jsl} describes
focused proofs for classical logic.  Its connections to polarization
and focusing were further explored and extended by Laurent, Quatrini
and de Falco \cite{laurent05apal} using {\em polarized proof nets.\/}
It may be tempting to speculate that the best way to arrive at a
notion of intuitionistic focusing is by simple modifications to these
systems, such as restricting them to single-conclusion sequents.
Closer examination however, reveal intricate issues concerning this
approach.  For example, the notion of classical {\em polarity\/}
appears to be distinct from and {\em contrary\/} to
intuitionistic polarity, especially at the level of atoms (see
Sections \ref{polarity} and \ref{classical}).  Resolving this issue
would be central to finding systems that support combined forward and
backward chaining.  
Although the relationship between $LK^\eta_p$ and our systems is
interesting, we chose for this work to derive intuitionistic focusing
from focusing in linear logic as opposed to classical logic.

Much of the research into focusing systems has been motivated by their
application.  For example, the papers
\cite{jagadeesan05fsttcs,miller91apal,howe98phd,chaudhuri06phd} are
motivated by foundational issues in logic programming and automated
deduction.  The papers
\cite{herbelin95phd,danos93wll,danos97jsl,laurent05apal} are motivated
by foundational issues in functional programming and the
$\lambda$-calculus.  Also, Levy \cite{levy06icalp}
presents focus-style proof systems for
typing in the $\lambda$-calculus and Curien and Herbelin
\cite{curien00icfp} (among others) have noted the relationship 
between forward chaining and call-by-value
evaluation and between backward chaining and call-by-name evaluation.

Our work can be extended to second order logic, although this
paper is concerned mainly with first-order quantification.

Many details missing from this paper can be found in the
report \cite{liang06report}.

\section{Focusing in Linear Logic}
\label{llfocus}

We summarize the key results from \cite{andreoli92jlc} on focusing
proofs for linear logic.

A {\em literal} is either an atomic formula or the negation of an
atomic formula.  A linear logic formula is in {\em negation normal
form} if it does not contain occurrences of $\limp$ and if all
negations have atomic scope.  If $K$ is literal, then $\nng{K}$
denotes its complement: in particular, if $K$ is $\nng{A}$ then
$\nng{K}$ is $A$.  

Connectives in linear logic are either {\em
asynchronous\/} or {\em synchronous}.  The asynchronous connectives are
$\bot$, $\lpar$, $\quest$, $\top$, $\with$, and $\forall$ while the
synchronous connectives are their de Morgan dual, namely, $\one$,
$\ot$, $\bang$, $\zero$, $\oplus$, and $\exists$.
Asynchronous connectives are
those where the right-introduction rule is always invertible.
Formally, a formula in negation normal form is of three kinds:
literal, asynchronous ({\em i.e.}, its top-level connective is 
asynchronous), and synchronous ({\em i.e.}, its top-level
connective is synchronous). 

\begin{figure}[ht]
\centering
$$\frac{\Up\Psi\Delta L}{\Up\Psi\Delta{\bot,L}}\ [\bot]
  \qquad
  \frac{\Up{\Psi}{\Delta}{F,G,L}}{\Up{\Psi}{\Delta}{F\lpar G,L}}\ [\lpar]
  \qquad
  \frac{\Up{\Psi,F}{\Delta}{L}}{\Up{\Psi}{\Delta}{\quest F,L}}\ [\quest]
$$
$$\frac{}{\Up{\Psi}{\Delta}{\top,L}}\ [\top]
  \qquad
  \frac{\Up{\Psi}{\Delta}{F,L}\quad \Up{\Psi}{\Delta}{G,L}}
       {\Up{\Psi}{\Delta}{F\with G,L}}
     \ [\with]
  \qquad
  \frac{\Up{\Psi}{\Delta}{B[y/x],L}}{\Up{\Psi}{\Delta}{\forall x.B,L}}\ 
       [\forall]
$$
$$\frac{\Up{\Psi}{\Delta,F}{L}}{\Up{\Psi}{\Delta}{F,L}}
   \ [R\Uparrow]\hbox{\quad provided that F is not asynchronous}
$$
$$\frac{}{\Down{\Psi}{\cdot}{\one}}\ [\one]
  \qquad
  \frac{\Down{\Psi}{\Delta_1}{F}\quad \Down{\Psi}{\Delta_2}{G}}
       {\Down{\Psi}{\Delta_1,\Delta_2}{F\ot G}} \ [\ot]
  \qquad
  \frac{\Up{\Psi}{\cdot}{F}}{\Down{\Psi}{\cdot}{\bang F}}\ [\bang]
$$
$$\frac{\Down{\Psi}{\Delta}{F_1}}{\Down{\Psi}{\Delta}{F_1\oplus F_2}}
  \ [\oplus_l]
  \qquad
  \frac{\Down{\Psi}{\Delta}{F_2}}{\Down{\Psi}{\Delta}{F_1\oplus F_2}}
  \ [\oplus_r]
  \qquad
  \frac{\Down{\Psi}{\Delta}{B[t/x]}}{\Down{\Psi}{\Delta}{\exists x.B}}
              \ [\exists]
$$
$$\frac{\Up{\Psi}{\Delta}{F}}{\Down{\Psi}{\Delta}{F}}
   \ [R\Downarrow]\hbox{\quad provided that $F$ is either asynchronous
                              or a negative literal} 
$$
$$\hbox{If $K$ a positive literal:\quad } 
  \frac{}{\Down{\Psi}{\nng{K}}{K}}\ [I_1]
  \qquad
  \frac{}{\Down{\Psi,\nng{K}}{\cdot}{K}}\ [I_2]
$$
$$\hbox{If $F$ is not a negative literal:\quad } 
  \frac{\Down{\Psi}{\Delta}{F}}{\Up{\Psi}{\Delta,F}{\cdot}}  \ [D_1]
  \qquad
   \frac{\Down{\Psi,F}{\Delta}{F}}{\Up{\Psi,F}{\Delta}{\cdot}} \ [D_2]
$$
\caption{The focused proof system LLF for linear logic}
\label{focused}
\end{figure}

As mentioned in Section~\ref{about}, 
the classification of non-atomic
formulas as asynchronous or synchronous is pushed to
literals by assigning a fixed but arbitrary {\em bias} to atoms: an
atom given a {\em negative bias} is linked to asynchronous behavior
while an atom given {\em positive bias} is linked to synchronous behavior.
In Andreoli's original
presentation of LLF \cite{andreoli92jlc} all atoms were classified as
``positive'' and their negations ``negative.'' 
Girard made a similar assignment for LC \cite{girard91mscs}.  In a classical
setting, such a choice works fine since classical negation simply
flips bias.  In intuitionistic systems, however, a more natural
treatment is to assign an arbitrary bias directly to atoms.
This bias of atoms is extended to literals: negating a negative atom
yields a positive literal and negating a positive atom yields a
negative literal.

The focusing proof system LLF for linear logic, presented in
Figure~\ref{focused}, contains two kinds of sequents.  In the sequent
$\Up\Psi\Delta L$, the ``zones'' $\Psi$ and $\Delta$ are multisets and
$L$ is a list.  This sequent encodes the usual one-sided sequent
$\oneseq\quest\Psi,\Delta, L$ (here, we assume the natural coercion of
lists into multisets).  This sequent will also satisfy the
invariant that requires $\Delta$ to contain only literals and
synchronous formulas.  In the sequent $\Down{\Psi}{\Delta}{F}$, the
zone $\Psi$ is a multiset of formulas and $\Delta$ is a multiset of
literals and synchronous formulas, and $F$ is a single formula.
Notice that the bias of literals is explicitly
referred to in the $[R\Uparrow]$ and initial rules: in particular, in
the initial rules, the literal on the right of the $\Downarrow$ must
be positive.

Changes to the bias assigned to atoms does not affect provability of a
linear logic formula: instead it affects the structure of focused
proofs.  

\section{Translating Intuitionistic Logic}
\label{trans}

\ignore{
\begin{table}[t]
$$\begin{tabular}
        {|@{\ }c@{\ \strut}||@{\strut\ }c@{\ }|@{\ \strut}c@{\ }|}
\hline
  $B$ & $B^1$ & $B^0$ \\
\hline\hline
atom $Q$           & $Q$ 
              & $Q$ \\
\hline
$\true$        & 1  
              & $\top$ \\
\hline
$\false$       & $0$ 
              & $0$\\
\hline
$P \wedge Q$  & $\bang (P^1 \& Q^1)$ 
              & $\bang P^0\&\bang
Q^0=(\quest\uob{P}\oplus\quest\uob{Q})^\bot$\\
\hline
$P \vee Q$    & $\bang P^1 \oplus \bang Q^1$ 
              & $\bang P^0 \oplus \bang Q^0 = (\quest \uob{P} \& \quest \uob{Q})^\bot$\\
\hline
$P\supset Q$  & $\bang (\bang P^0 \limp Q^1) = \bang (\quest \uob{P} \lpar Q^1) $ 
              & $\bang P^1 \limp \bang Q^0 = (\bang P^1 \otimes \qqob)^\bot$\\
\hline
$\neg P$      & $\bang (\bang P^0\limp 0) = \bang (0\lpar \quest \uob{P}) $ 
              & $\bang P^1\limp 0 = (\bang P^1 \otimes \top)^\bot$\\
\hline
$\exists x P$ & $\exists x \bang P^1$ 
              & $\exists x \bang P^0 = (\forall x\qob{P})^\bot$\\
\hline
$\forall x P$ & $\bang \forall x P^1$ 
              & $\forall x \bang  P^0 = (\exists x\qob{P})^\bot$\\
\hline
 \end{tabular}
$$
\caption{The 0/1 translation used to encode LJ proofs into linear
  logic.}
\label{01trans}
\end{table}
}
\begin{table}[t]
$$\begin{tabular}
        {|@{\ }c@{\ \copy\tallstrutbox}||@{\ }c@{\ }|@{\ }c@{\ }|@{\ }c@{\ }|}
\hline
  $B$ & $B^1$ & $B^0$ & $(B^0)^\perp$ \\
\hline\hline
atom $Q$      & $Q$ & $Q$ & $Q^\perp$ \\
\hline
$\true$       & 1 & $\top$ & $0$ \\
\hline
$\false$      & $0$ & $0$  & $\top$ \\
\hline
$P \wedge Q$  & $\bang (P^1 \& Q^1)$ 
              & $\bang P^0\&\bang Q^0$
              & $\quest\uob{P}\oplus\quest\uob{Q}$\\
\hline
$P \vee Q$    & $\bang P^1 \oplus \bang Q^1$ 
              & $\bang P^0 \oplus \bang Q^0$ 
              & $\quest \uob{P} \& \quest \uob{Q}$\\
\hline
$P\supset Q$  & $\bang (\quest \uob{P} \lpar Q^1)$ 
              & $\bang P^1 \limp \bang Q^0$
              & $\bang P^1 \otimes \quest (Q^0)^\bot$\\
\hline
$\neg P$      & $\bang (0\lpar \quest \uob{P}) $ 
              & $\bang P^1\limp 0$ 
              & $\bang P^1 \otimes \top$\\
\hline
$\exists x P$ & $\exists x \bang P^1$ 
              & $\exists x \bang P^0$
              & $\forall x\qob{P}$\\
\hline
$\forall x P$ & $\bang \forall x P^1$ 
              & $\forall x \bang P^0$
              & $\exists x\qob{P}$\\
\hline
 \end{tabular}
$$
\caption{The 0/1 translation used to encode LJ proofs into linear
  logic.}
\label{01trans}
\end{table}

Table \ref{01trans} contains a translation of intuitionistic logic
into linear logic.  This translation induces a bijection between
arbitrary LJ proofs and LLF proofs of the translated image in the
following sense.  First notice that this translation is {\em
asymmetric}: the intuitionistic formula $A$ is translated using $A^1$
if it occurs on the right-side of an LJ sequent and as $A^0$ if it
occurs on the left-side.  Since this translation is used
to capture cut-free proofs, such distinctions are not problematic.
Since the left-hand side of a sequent in LJ will be negated when
translated to a one-sided linear logic sequent, $(B^0)^\bot$ is
also shown.
The following Proposition essentially says that, via this translation,
linear logic focusing can capture arbitrary proofs in LJ.
Here, $\vdashi$ denotes an intuitionistic logic sequent.
\begin{proposition}
Let $\uob{\Gamma}$ be the multiset $\{(D^0)^\perp\mid D\in\Gamma\}$.
The focused proofs of $\fgob:\Uparrow R^1$ are in
bijective correspondence with the LJ proofs of $\Gamma\vdashi R$.
\end{proposition}

For a detailed proof, see \cite{liang06report}.  We simply illustrate
one case in constructing the mapping from a collection of LLF rules to
an LJ inference rule. 
\[
\vcenter{
\infer[\lbrack\oplus\rbrack]{\fgob :R^1\Downarrow \qob{D_1}\oplus \qob{D_2}}
{\infer[\lbrack R\Downarrow\rbrack]{\fgob :R^1\Downarrow \qob{D_i}}
 {\infer[\lbrack?\rbrack]{\fgob :R^1\Uparrow \qob{D_i}}
  {\fgob, \uob{D_i} : R^1 \Uparrow}}}}
\qquad  \longmapsto \qquad 
\vcenter{\infer[\lbrack\wedge L\rbrack]{\Gamma, D_1\wedge D_2\vdashi R}
{\Gamma, D_i \vdashi R}}
\]
The liberal use of $\bang$ in this translation {\em throttles\/} focusing.
This translation is reminiscent of the earliest embedding
of classical into intuitionistic logic of Kolmogorov, which uses the
double negation in a similarly liberal fashion.  

The 0/1 translation can be used as a starting point in establishing
the completeness of other proof systems.  These systems can be seen as
{\em induced\/} from alternative translations of intuitionistic logic.

\begin{wrapfigure}{r}{55mm}
\begin{tabular}{|c||c|c|@{\copy\tallstrutbox}}
\hline
$F$ & $F^q$ (right) & $F^j$ (left) \\
\hline\hline
atom $C$ & $C$ & $C$  \\
\hline
$\false$ & $0$ & $0$ \\
\hline
$A\wedge B$ & $A^q\otimes B^q$ & $\bang A^j\otimes \bang B^j$ \\
\hline
$A\vee B$ & $A^q\oplus B^q$ & $\bang A^j\oplus \bang B^j$\\
\hline
$A\supset B$ & $(\bang A^j \lolli B^q)\otimes 1$ & \ $A^q\lolli \bang B^j$\ \\
\hline
\end{tabular} 
\end{wrapfigure}
Consider, for example, the \LJQp\ proof system presented in
\cite{dyckhoff06cie}.  The translation for this system is given here
using the ``$q/j$'' mapping.  The ``$\otimes 1$'' device is another
way to control focusing, or the lack thereof.  {\em All atoms must be
given positive bias\/} for this translation.

With minor changes, Girard's original (non-polarized) translation of
intuitionistic logic \cite{girard87tcs} induces the complement to
\LJQp\ called {\em LJT\/} \cite{herbelin95phd}, which is itself
derived from LKT \cite{danos93wll} (where this connection was noted
implicitly.)  {\em All atoms must be given negative bias\/} for this
translation.

Given a translation such as that of \LJQp, one can give a completeness
proof for the system using a {\em ``grand tour''\/} through linear logic
as follows:
\begin{enumerate}
\item Show that a proof under the 0/1 translation can be converted
into a proof under the new translation.  This usually follows 
from cut-elimination.  

\item Define a mapping between proofs in the new system (such as LJQ) 
and LLF proofs of its translation.

\item Show soundness of the new system with respect to LJ.  This is usually
trivial.  The ``tour'' is now complete, since proofs in LJ map to
proofs under the 0/1 translation.
\end{enumerate}

An intuitionistic system that contains atoms of both positive and
negative bias is $\lambda$RCC \cite{jagadeesan05fsttcs}.  Two special
cases of the $\supset\!L$ rule are distinguished involving $E\supset
D$ for positive atom $E$ and $G\supset A$ for negative atom $A$. Each
rule requires that the complementary atom ($E$ on the left, $A$ on the
right) is present when applied, thus terminating one branch of the
proof.  One can translate these special cases using forms $E\limp
\bang D'$ and $\bang G' \limp A$, respectively, in linear logic.  The
strategy outlined above can then be used to not only prove its
completeness but also extend it with more aggressive focusing
features.

Our interest here is not the construction of individual systems but
the building of a unifying framework for focusing in intuitionistic logic.
Such a task requires a closer examination of {\em polarity\/} and its
connection to focusing.

\section{Permeable Formulas and their Polarity}\label{polarity}

Focused proofs in linear logic are characterized by two different
phases: the invertible (asynchronous) phase and the non-invertible
(synchronous) phase.  These two phases are characterized by
introduction rules for dual sets of formulas.  In order to construct a
general focusing scheme for intuitionistic logic, the non-linear
(exponential) aspects of proofs need special attention, especially in
light of the fact that the $[\bang]$ rule stops a bottom-up
construction of focused application of synchronous rules (the arrow
$\Downarrow$ in the conclusion flips to $\Uparrow$ in the premise).

For our purposes here, a particularly flexible way to deal with the
exponentials in the translations of intuitionistic formulas is via the
notion of {\em permeation} that is used in LU
\cite{girard93apal}. In particular, there are essentially three grades
of {\em permeation}.  The formula $B$ is {\em left-permeable} if
$B\equiv \bang B$, is {\em right-permeable} if $B\equiv \quest B$, and
{\em neutral} otherwise.  Within sequent calculus proofs, a formula is
left-permeable if it admits structural rules on the left and
right-permeable if it admits structural rules on the right.  An
example of a left-permeable formula is $\exists x \bang\!A$.
All left-permeable formulas are synchronous and all right-permeables
asynchronous.  In the LU system, both the
left and right sides of sequents contain two zones --- one that treats
formulas linearly and one that permits structural rules.  A
left-permeable (resp., right-permeable) formula is allowed to move
between both zones on the left (right).
\ignore{  
**** It is important to explicate the precise relationship between the
LU polarization scheme and the synchronous/asynchronous duality as
defined by Andreoli.  On the surface, there are several discrepancies
between these notions.  There are three polarities in LU.  {\em
  Positives\/} are those formulas that admit structural rules on the
left: in linear logic formulas $A$ such that $A\equiv \bang A$.
Dually, {\em negatives\/} are those such that $A\equiv \quest A$.
{\em Neutrals\/} are those that do not admit structural rules
anywhere. LU uses a dyadic system similar to LLF: explicit
exponentials are replaced by separate linear and unbounded contexts or
zones.  Movement between zones is called {\em permeation.\/}
Dereliction is a permeation rule, but the positive and negative
formulas can also permeate in the other direction: left-permeable for
positives, right-permeable for negatives.  An example of a
left-permeable (positive) formula is $\exists x \bang A$.  All
left-permeable formulas are synchronous and all right-permeables
asynchronous (but not vice versa).  
} 
In addition, LU introduces 
atoms that are inherently left or right-permeable or neutral.
Although they appear to properly extend linear logic, one can
simulate LU in ``regular'' linear logic by translating left-permeable atoms
$A$ as $\bang A$ and right-permeable ones as $\quest A$.

To preserve the focusing characteristics of permeable atoms as
positively or negatively biased atoms, we use the following LU-inspired
asymmetrical translation.  The superscript $-\!1$ indicates
the left-side translation and $+\!1$ indicates the right-side
translation:
\begin{quote}
$P^\none=\bang P$ and $P^\pone = P$, for left-permeable (positive) atom $P$.\\
$N^\none=N$ and $N^\pone=\quest N$, for right-permeable (negative) atom $N$.\\
$B^\none=B^\pone = B$,  for neutral atom $B$.
\end{quote}
The $\bang$ rule of LLF causes a loss of focus in all circumstances,
and is the main reason why we use an asymmetrical translation.  The
translation of positive atoms above preserves {\em permeation on the
left\/} while allowing for {\em focus on the right.\/} That is,
left-permeable atoms can now be interpreted meaningfully as positively
biased atoms in focused proofs, and dually for right-permeable atoms.
Furthermore, the permeation of positive atoms is {\em ``one-way
only:''\/} they cannot be selected for focus again once they enter the
non-linear context.

Intuitionistic logic uses the left-permeable and neutral formulas and atoms.
LU defines a translation for intuitionistic logic so that {\em all
synchronous formulas are left-permeable.\/} For example, $\vee$ is
translated as follows (here, $P$, $Q$ are positive and $N$, $M$ are negative):
$(P\vee Q)^\none = P^\none \oplus Q^\none$, 
$(P\vee N)^\none = P^\none \oplus \bang N^\none$, 
$(N\vee P)^\none = \bang N^\none \oplus Q^\none$, and 
$(N\vee M)^\none = \bang N^\none \oplus \bang M^\none$.
The final element of intuitionistic polarity is that {\em
neutral atoms should be assigned negative bias in focused proofs.\/}
Neutral atoms that are introduced into the left context (e.g. by a
$\supset\!L$ rule) must immediately end that branch of the proof in an
identity rule.  Otherwise, the unique {\em stoup\/} is lost when
multiple non-permeable atoms accumulate in the linear context.

The LU and LLF systems serve as a convenient platform
for the unified characterization of polarity and focusing in all three
logics.  We can now understand the terminology of ``positive'' and
``negative'' formulas in each logic as follows:

\ignore{
\begin{enumerate}
\item Positive polarity in linear logic means synchronous.
\item Negative polarity in linear logic means asynchronous.
\item Positive polarity in intuitionistic and classical logic
means synchronous {\em and\/} left-permeable in linear logic.
\item Negative polarity in intuitionistic logic means asynchronous
{\em and\/} non-permeable in linear logic.
\item Negative polarity in classical logic means asynchronous {\em and\/}
right-permeable in linear logic.
\end{enumerate}
Furthermore, the notion of positive/negative in each logic can be
extended to atoms, with the consequence that some atoms must be made
permeable. 
}

\begin{description}
\item[Linear logic:] {\em Positive\/} formulas are synchronous
formulas and positively biased neutral atoms.  {\em Negative\/}
formulas are asynchronous formulas and negatively biased neutral 
atoms. 

\item[Intuitionistic logic:] {\em Positive\/} formulas
are left-permeable formulas and
{\em negative\/} formulas are asynchronous
neutral formulas and negatively biased neutral atoms.

\item[Classical logic:] {\em Positive\/} formulas are
left-permeable formulas.
{\em Negative\/} formulas are right-permeable formulas.
\end{description}

\section{The \ljf\ Sequent Calculus}
\label{ljf}

Since the polarities of intuitionistic logic observe stronger
invariances, intuitionistic focused proofs are more
well-structured than LLF proofs.  The non-linear context of LLF
contains both synchronous and asynchronous formulas, whereas in
intuitionistic logic sequents can be clearly divided into zones
respecting polarity. That is, when translating an intuitionistic
sequent into a LLF sequent, synchronous formulas on the left are
placed in the linear context.  

We also make an adjustment on
the LU translation of intuitionistic logic.  Instead of using $\&$ or
$\otimes$ depending on the polarities of the subformulas, we 
construct two versions of intuitionistic conjunction, which
has the following meaning in linear logic ($P$, $Q$ for positives,
$N$, $M$ for negatives, $A$, $B$ arbitrary):
\begin{quote}
\begin{tabular}{l@{\qquad}l}
$(P\wedgep Q)^\none = P^\none \otimes Q^\none$ &
$(A \wedgep B)^\pone  =  A^\pone \otimes B^\pone$\\
$(P\wedgep N)^\none  =  P^\none \otimes \bang N^\none$ & $~$\\
$(N\wedgep P)^\none  =  \bang N^\none \otimes P^\none$ &
$(A \wedgen B)^\none  =  A^\none \& B^\none$\\
$(N\wedgep M)^\none  =  \bang N^\none \otimes \bang M^\none$ &
$(A \wedgen B)^\pone  =  A^\pone \& B^\pone$
\end{tabular}
\end{quote}
The connectives $\wedgen$ and $\wedgep$ are equivalent in
intuitionistic logic in terms 
of provability but differ in their impact on the structure of focused proofs.
The use of two conjunctions means that 
the top-level structure of formulas completely determines their polarity.  
{\em Polarity\/} in intuitionistic logic is defined as follows.
\begin{definition}\label{ljfpolarity}
Atoms in LJF are arbitrarily positive or negative.  {\em Positive
formulas\/} are among positive atoms, $\true$, $\false$,
$A \wedgep B$, $A\vee B$ and $\exists x A$.  
{\em Negative formulas\/} are 
among negative atoms, $A \wedgen B$, $A \supset B$ and $\forall x A$.
\end{definition}

The above translation induces the
sequent calculus {\em LJF\/} for intuitionistic logic, shown in Figure
\ref{ljf2}.  Sequents in \ljf\ can be interpreted as follows:
\begin{enumerate}
\item $[\Gamma],\Theta\lra {\cal R}$ \ (end sequent): \ 
this is an {\em unfocused sequent.\/}  $\Gamma$ contains
negative formulas and positive atoms.  
$\cal R$ represents either a formula $R$ or $[R]$.
\item $[\Gamma]\lra [R]$: \ this represents a sequent in which all asynchronous
formulas have been decomposed, and is ready for a formula 
to be selected for focus.
\item $[\Gamma]\leftfcs{A} [R]$: \ this is a {\em left-focusing\/} sequent,
with focus on formula $A$.  The meaning of this sequent remains
$\Gamma,A\vdashi R$.
\item $[\Gamma]\rightfcs{A}$: \  this is a {\em right-focusing\/} sequent on
formula $A$, with the meaning $\Gamma\vdashi A$.
\end{enumerate} 
\ignore{ 
The following is a sample correspondence between a LLF rule and a \ljf\ rule:
\[
\vcenter{\infer[\!\invamp]{\fgnb :\Uparrow R^\pone, \unb{\Theta},\unb{A}\lpar \unb{B}}
{\fgnb :\Uparrow R^\pone, \unb{\Theta},\unb{A}, \unb{B}}}
~~\longmapsto~~
\vcenter{\infer[\wedgep L]{[\Gamma],\Theta,A\wedgep B\lra R}
  {[\Gamma],\Theta,A, B\lra R}}
\]
} 

\begin{figure}[t]
\[
\infer[Lf]{[N,\Gamma]\lra [R]}{[N,\Gamma]\leftfcs{N} [R]}
\qquad
\infer[Rf]{[\Gamma]\lra [P]}{[\Gamma]\rightfcs{P}}
\qquad
\infer[R_l]{[\Gamma]\leftfcs{P} [R]}{[\Gamma],P\lra [R]}
\qquad
\infer[R_r]{[\Gamma]\rightfcs{N}}{[\Gamma]\lra N}
\]
\[
\infer[\mbox{[]}_l]{[\Gamma],\Theta, C\lra {\cal R}}{[C,\Gamma],\Theta\lra {\cal R}}
\qquad\qquad
\infer[\mbox{[]}_r]{[\Gamma],\Theta \lra D}{[\Gamma],\Theta\lra [D]}
\]
\[
\infer[I_r, ~\mbox{atomic $P$}]{[P,\Gamma]\rightfcs{P}}{}
\qquad\qquad
\infer[I_l, ~\mbox{atomic $N$}]{[\Gamma]\leftfcs{N} [N]}{}
\]
\[
\infer[falseL]{[\Gamma],\Theta, false\lra {\cal R}}{}
\qquad
\infer[trueL]{[\Gamma],\Theta,true\lra {\cal R}}{[\Gamma],\Theta\lra {\cal R}}
\qquad
\infer[trueR]{[\Gamma] \rightfcs{true}}{}
\]
\[
\infer[\wedgen L]{[\Gamma]\leftfcs{A_1 \wedgen A_2}  [R]}{[\Gamma]\leftfcs{A_i} [R]}
\qquad\qquad
\infer[\wedgep L]{[\Gamma],\Theta,A ~\wedgep B\lra {\cal R}}
{[\Gamma],\Theta,A,B\lra {\cal R}}
\]
\[
\infer[\wedgen R]{[\Gamma],\Theta\lra A ~\wedgen B}
{[\Gamma],\Theta\lra A & [\Gamma],\Theta\lra B}
\qquad\qquad
\infer[\wedgep R]{[\Gamma]\rightfcs{A ~\wedgep B}}
{[\Gamma]\rightfcs{A} & [\Gamma]\rightfcs{B}}
\]
\[
\infer[\vee L]{[\Gamma],\Theta, A\vee B\lra {\cal R}}
{[\Gamma],\Theta,A\lra {\cal R} & [\Gamma],\Theta,B\lra {\cal R}}
\qquad~~
\infer[\vee R]{[\Gamma]\rightfcs{A_1\vee A_2}}{[\Gamma]\rightfcs{A_i}}
\]
\[
\infer[\supset L]{[\Gamma]\leftfcs{A\supset B} [R]}
 {[\Gamma]\rightfcs{A}~ & ~[\Gamma]\leftfcs{B} [R]}
\qquad\qquad
\infer[\supset R]{[\Gamma],\Theta\lra A\supset B}{[\Gamma],\Theta,A\lra B}
\]
\[
\infer[\exists L]{[\Gamma],\Theta,\exists y A\lra {\cal R}}
                 {[\Gamma],\Theta,A\lra {\cal R}}
\quad\
\infer[\exists R]{[\Gamma]\rightfcs{\exists x A}}{[\Gamma]\rightfcs{A[t/x]}}
\qquad
\infer[\forall L]{[\Gamma]\leftfcs{\forall x A} [R]}{[\Gamma]\leftfcs{A[t/x]} [R]}
\quad\
\infer[\forall R]{[\Gamma],\Theta\lra \forall y A}{[\Gamma],\Theta\lra A}
\]
\caption{The Intuitionistic Sequent Calculus LJF.  Here, $P$ is
  positive, $N$ is negative, $C$ is a negative formula or
  positive atom,  and $D$ a positive formula or negative atom.  Other
  formulas are arbitrary.  Also, $y$ is not free in $\Gamma$,
  $\Theta$, or {\cal R}. 
}
\label{ljf2}
\end{figure}

\begin{theorem} LJF is sound and complete with respect to intuitionistic logic.
\end{theorem}
\begin{proof}
Using the ``grand tour'' strategy.  See \cite[Section
  6]{liang06report} for details. 
\end{proof}

Given the different forms of sequents, the cut rule for \ljf\ takes
many forms:
\[
\infer[Cut^+]{[\Gamma\Gamma' ],\Theta\Theta' \lra \cal R}
{[\Gamma],\Theta \lra P & [\Gamma'],\Theta',P\lra \cal R}
\quad
\infer[Cut^-]{[\Gamma\Gamma' ],\Theta\Theta' \lra \cal R}
{[\Gamma],\Theta \lra C & [C, \Gamma'],\Theta' \lra \cal R}
\]
\[
\infer[Cut_1^\leftarrow]{[\Gamma\Gamma']\leftfcs{B} [R]}
{[\Gamma]\leftfcs{B} [P] & [\Gamma'],P \lra [R]}
\quad
\infer[Cut_2^\leftarrow]{[\Gamma\Gamma']\leftfcs{B} [R]}
{[\Gamma]\lra N & [N,\Gamma']\leftfcs{B} [R]}
\]
\[
\infer[Cut^\ra]{[\Gamma\Gamma']\rightfcs{R} }
{[\Gamma]\rightfcs{C} & [C,\Gamma']\rightfcs{R}}
\]
Notice that the last three cut rules retain focus in the conclusion.
These rules extend those of {\em LJQ'\/} \cite{dyckhoff06cie}, which
were shown to be useful for studying term-reduction systems.
See \cite{liang06report} for a proof of the admissibility of these
rules. 

Like LLF, a key characteristic of \ljf\ is the assignment of arbitrary
polarity to atoms.  To illustrate the effect of these assignments on
the structure of focused proofs, consider the sequent $a, a\supset b,
b\supset c \vdash c$ where $a$, $b$ and $c$ are atoms.  This sequent
can be proved either by {\em forward chaining\/} through the clause
$a\supset b$, or {\em backward chaining\/} through the clause
$b\supset c$.  Assume that atoms $a$ and $b$ are assigned positive
polarity and that $c$ is assigned negative polarity.  This assignment
effectively adopts the forward chaining strategy, reflected in the following \ljf\
proof segment (here, $\Gamma$ is the set $\{a,a\supset b, b\supset c\}$):
\[
\infer[\supset\!L]{[\Gamma]\leftfcs{a\supset b} [c]}
{ \infer[I_r]{[\Gamma]\rightfcs{a}}{}
  &
  \infer[R_l]{[\Gamma]\leftfcs{b} [c]}
  {\infer[\boxl]{[\Gamma], b\lra [c]}
   {\infer[Lf]{[b,\Gamma]\lra [c]}
    {\infer[\supset\!L]{[b,\Gamma]\leftfcs{b\supset c} [c]}
    {\infer[I_r\quad]{[b,\Gamma]\rightfcs{b}}{} 
     &
     \infer[I_l]{[b,\Gamma]\leftfcs{c} [c]}{}}}}}}
\]
The polarities of $a$ and $c$ do not fundamentally affect the structure
of the proof in this example.  However,
assigning negative polarity to atom $b$ would restrict the proof to
use the backward chaining strategy:
\[
\infer[\supset\!L]{[\Gamma]\leftfcs{b\supset c}[c]}
{\infer[R_r]{[\Gamma]\rightfcs{b}}
 {\infer[\boxr]{[\Gamma]\lra b}
  {\infer[Lf]{[\Gamma]\lra [b]}
   {\infer[\supset\!L]{[\Gamma]\leftfcs{a\supset b} [b]}
    {\infer[I_r\quad]{[\Gamma]\rightfcs{a}}{}
     &
     \infer[I_l]{[\Gamma]\leftfcs{b} [b]}{}}}}}
 & \infer[I_l]{[\Gamma]\leftfcs{c} [c]}{}}
\]

\section{Embedding Intuitionistic Systems in LJF}
\label{framework}

The \ljf\ proof system can be used to ``host'' other focusing proof
system for intuitionistic logic.  One obvious restriction to \ljf\ is
its purely negative fragment, which essentially corresponds to LJT.  In the
negative fragment one also finds {\em uniform proofs,\/} where the
right ``goal'' formula is always fully decomposed before any left rule
is applied.  Various other proof systems can be embedded into \ljf\ by
mapping intuitionistic formulas to intuitionistic formulas
in such a way that focusing features in \ljf\ are stopped by the
insertion of {\em delay\/} operators.  In particular, if we define
$\mkneg{B}= true\supset B$ and $\mkpos{B}=true\wedgep B$, then $B$,
$\mkneg{B}$, and $\mkpos{B}$ are all logically equivalent but
$\mkneg{B}$ is always negative and $\mkpos{B}$ is always positive.

Proofs in the LJQ' system can be embedded into \ljf\ by translating
all left-side formulas ($l$) as negatives and all right-side formulas
($r$) as positives: in particular, for atom B, $B^l = B^r = B$,
$false^l = \mkneg{false}$ , \  $false^r = false$,
$(A\wedge B)^l = \mkneg{A^l\wedgep B^l}$,  
      $(A\wedge B)^r = A^r\wedgep B^r$,
$(A\vee B)^l = \mkneg{A^l\vee B^l}$,   
      $(A\vee B)^r = A^r \vee B^r$,
$(A\supset B)^l = A^r\supset \mkpos{B^l}$, 
      $(A\supset B)^r = \mkpos{A^l \supset B^r}$.

\begin{wrapfigure}{r}{70mm}
\vskip -18pt
\begin{tabular}{@{\vrule height9.8pt depth3.5pt width0pt}|c||c|c|}
\hline
$F$ & $F^l$ (left) & \ \ \ $F^r$ (right) \\
\hline
\hline
atom $C$ & $C$ & $C$\\
\hline
$\false$ & $\mkneg{\false}$ & $\false$\\
\hline
$\true$ & $\mkneg{\true}$ & $\true$\\
\hline
$A\wedge B$ & $\mkpos{A^l}\wedgen \mkpos{B^l}$ & $\mkpos{A^r\wedgen B^r}$\\
\hline
$A\vee B$ & $\mkneg{A^l \vee B^l}$ & $\mkneg{A^r} \vee \mkneg{B^r}$\\
\hline
$A\supset B$ & $\mkneg{A^r} \supset \mkpos{B^l}$ & $\mkpos{A^l\supset B^r}$\\
\hline
$\exists x A$ & $\mkneg{\exists x A^l}$ & $\exists x \mkneg{A^r}$\\
\hline
$\forall x A$ & $\forall x \mkpos{A^l}$ & $\mkpos{\forall x A^r}$\\
\hline
\end{tabular} 
\end{wrapfigure}
Arbitrary LJ proofs can be embedded within \ljf\ by inserting sufficient
delaying operators.  The table here provides the translation
(redefining the superscripts $l$ and $r$, for convenience).  Together
with cut-elimination, the embedding also suggests a
completeness proof for \ljf\ independently of linear logic.
The following example embeds the $\wedge R$ rule in \ljf:
\[
\infer[Rf]{[\Gamma]\lra [\mkpos{A^r\wedgen B^r}]}
{\infer[\wedgep R]{[\Gamma]\rightfcs{(A^r\wedgen B^r)\wedgep \true}}
{\infer[R_r]{[\Gamma]\rightfcs{A^r\wedgen B^r}}
  {\infer[\wedgen R]{[\Gamma]\lra A^r\wedgen B^r}
     { \infer[\boxr]{[\Gamma]\lra A^r}{[\Gamma]\lra [A^r]} &
       \infer[\boxr]{[\Gamma]\lra B^r}{[\Gamma]\lra [B^r]} }}
   & \infer[\true R]{[\Gamma]\rightfcs{\true}}{} }}
\]
The system $\lambda$RCC also presents interesting choices.  In
particular, it may not always be the best choice to focus
maximally.  Forward chaining may generate a new formula or ``clause''
that may need to be used multiple times.  In a $\supset\!L$ rule on
formulas $E\supset D$ where $E$ is a positive atom, one may not wish
to decompose the formula $D$ immediately.  This is accomplished in the
linear translation with a $\bang$.  It can also be accomplished by
using formulas $E\supset \mkpos{D}$ in case $D$ is negative, and
 $E\supset \mkpos{\mkneg{D}}$ in case $D$ is positive.  
Note that unlike the $l/r$ translations for LJQ and LJ above, these simple
devices do not hereditarily alter the structure of $D$.

\section{Embedding Classical Logic in LJF}
\label{classical}

We can use \ljf\  to formulate a focused sequent calculus for
classical logic that reveals the latter's constructive content in the
style of LC.  While it is possible to derive such a system again using
linear logic, classical logic can also be embedded within
intuitionistic logic using the {\em double-negation\/}
translations of G\"odel, Gentzen, and Kolmogorov.  These translations do
not, however, yield significant focusing features.
Girard's {\em polarized\/} version of the double negation translation
for LC approaches the problem of capturing duality in a more subtle
way.  Following the style of \ljf, we wish to define dual
versions of each propositional connective, which leads to a more
usable calculus.  We thus modify the LC translation in a natural way,
which is consistent with its original intent.  The proof system we derive is
called {\em LKF}.

We must first separate classical from intuitionistic polarity since
these are different notions (see the end of Section \ref{polarity}).
\begin{definition}\label{lkfpolarity}
Atoms are arbitrarily classified as either positive or negative.
The literal $\neg A$ has the opposite polarity of the atom $A$.
{\em Positive formulas\/} are among positive literals, $\cal T$, $\cal F$, 
$A\wedgep B$, $A\veep B$, $A\implp B$ and $\exists x A$. {\em Negative
formulas\/} are among negative literals, $\neg {\cal T}$, $\neg {\cal F}$,
$A\wedgen B$, $A\veen B$, $A\impln B$ and $\forall x A$.
Negation $\neg A$ is defined by de Morgan dualities\  $\neg A/A$,
$\wedgep/\veen$, $\wedgen/\veep$ and $\forall/\exists$.
Negative implication $A\impln B$ is defined as $\neg A \veen B$ and
$A\implp B$ is defined as $\neg A \veep B$.  
Formulas are assumed to be in negation normal form (that is,
formulas that do not contain implications and negations have
atomic scope).
\end{definition}

The constants $\cal T$, $\cal F$, $\neg \cal T$ and $\neg \cal F$ are
best described, respectively, as $1$, $0$, $\bot$ and $\top$ in linear
logic.  Just as we have dual versions of each connective, we also have
dual versions of each identity.  But this is not linear logic as the
formulas are polarized {\em in the extreme.\/}  The distinction between
the positive and negative versions of each connective affects only the
structure of proofs and not provability.

\begin{table}[ht]
\begin{center}
\begin{tabular}{|r|r||c|c|c|c|c|}
\hline
$\ddp{\cal A}$ & $\ddp{\cal B}$ & $\ddp{({\cal A}\wedgep {\cal B})}$ & 
$\ddp{({\cal A}\wedgen {\cal B})}$ & $\ddp{({\cal A}\veep {\cal B})}$ & 
$\ddp{({\cal A}\veen {\cal B})}$ & $\ddp{(\neg {\cal A})}$\\
\hline \hline
$A$ & $B$ & $A\wedgep   B$ & $\negg(\negg A \vee \negg B)$ &
$  A\vee   B$ & $\negg(\negg  A \wedgep \negg  B)$ & 
$\negg A$\\
\hline
$A$ & $\negg B$ & $  A \wedgep \negg B$ & $ \negg (\negg A \vee B)$ &
$  A \vee \negg B$ & $\negg (\negg A \wedgep   B)$ & $\cdot$\\
\hline
$\negg A$ & $B$ & $\negg A\wedgep   B$ & $\negg (A \vee \negg B)$ &
$\negg A \vee   B$ & $\negg (  A \wedgep \negg B)$ & $  A$\\
\hline
$\negg A$ & $\negg B$ & $\negg A \wedgep \negg B$ & 
$\negg (A\vee B)$ &
$\negg A \vee \negg B$ & $\negg (  A \wedgep   B)$ & $\cdot$\\
\hline
\end{tabular} 

\begin{tabular}{|r|r||c|c|c|c|@{\copy\tallstrutbox}}
\hline
$\ddp{\cal {\cal A}}$ & $\ddp{\cal B}$ & $\ddp{({\cal A}\implp {\cal B})}$ &
  $\ddp{({\cal A}\impln {\cal B})}$ 
& $\ddp{(\forall x {\cal A})}$ & $\ddp{(\exists x {\cal A})}$  \\
\hline \hline
$A$ & $B$ & $\negg  A \vee   B$ & $\negg (  A \wedgep \negg   B)$
& $\negg (\exists x \negg A)$ & $\exists x   A$ \\
\hline
$A$ & $\negg B$ & $\negg A \vee \negg B$ & $\negg (  A\wedgep   B)$ 
& $\cdot$ & $\cdot$\\
\hline
$\negg A$ & $B$ & $  A\vee   B$ & $\negg (\negg A\wedgep \negg B)$
& $\negg (\exists x   A)$ & $\exists x \negg A$\\
\hline
$\negg A$ & $\negg B$ & $  A\vee \negg B$ & $\negg (\negg A\wedgep   B)$ 
& $\cdot$ & $\cdot$\\
\hline
\end{tabular} 
\end{center}
\caption{Polarized embedding of classical logic.  
The $\ddp{(\cdot)}$ translation on compound formulas is given above
(there, $A$, $B$ represent formulas not preceded by $\sim$).  For
positive classical atom $P$, $\ddp{P} = P$; 
for negative classical atom $N$, $\ddp{N} = \negg N$; (where both  $P$
and $N$ are assigned positive intuitionistic polarity), and for the
logical constants $\ddp{\cal T} = \true$, $\ddp{\cal F} = \false$,
$\ddp{(\neg \cal T)} = \negg \true$, $\ddp{(\neg \cal F)} = \negg
\false$.}
\label{dnegtable}
\end{table}

Let $\negg\!A$ represent the intuitionistic formula $A\supset \phi$
where $\phi$ is {\em some unspecified positive atom.\/} The
``$\approx$'' embedding of classical logic is found in Table
\ref{dnegtable}.
Variations are possible on the embedding. Note that the classical
$\wedgen$ is not defined in terms of the intuitionistic $\wedgen$. The
embeddings are selected to enforce the dualities $\wedgen/\veep$ and
$\wedgep/\veen$.  Alternatives may also work, but will increase the
complexity of the derivation.  Here, the cases all follow
the pattern $P$ or $\negg P$ where $P$ is a positive intuitionistic
formula.  In particular, negative intuitionistic atoms are
not used in the embedding.

The $\approx$ embedding induces the {\em LKF\/} sequent calculus in
Figure~\ref{lkf} from the image of \ljf\ proofs, analogous to how
\ljf\ was derived from LLF.  
Here is one sample correspondence between a \ljf\ rule and a LKF rule:
\[
\vcenter{\infer[\wedgep L]{[\Delta],\Psi,  A\wedgep   B\lra [\phi]}
{[\Delta],\Psi,  A,   B\lra [\phi]}}
\qquad \longmapsto \qquad
\vcenter{\infer[\veen]{\vdash [\Theta],\Gamma, A \veen B}
{\vdash [\Theta],\Gamma, A, B}}
\]
Sequents of the form $\vdash [\Theta],\Gamma$ are unfocused while
those of the form $\mapsto [\Theta],A$ focus on the {\em stoup\/} formula $A$.

\begin{figure}[t]
\[
\infer[\mbox{$[]$}]{\vdash [\Theta], \Gamma, C}{\vdash [\Theta, C],\Gamma}
\qquad
\infer[Focus]{\vdash [P,\Theta]}{\mapsto [P,\Theta],P}
\qquad
\infer[Release]{\mapsto [\Theta], N}{\vdash [\Theta] , N}
\]
\[
\infer[ID^+, ~\mbox{atomic $P$}]{\mapsto [\neg P,\Theta] , P}{}
\qquad\qquad
\infer[ID^-, ~\mbox{atomic $N$}]{\mapsto [N,\Theta], \neg N}{}
\]
\[
\infer[indeed]{\mapsto [\Theta],{\cal T}}{}
\qquad
\infer[absurd]{\vdash [\Theta],\Gamma, \neg {\cal F}}{}
\qquad
\infer[trivial]{\vdash [\Theta],\Gamma, \neg {\cal T}}
{\vdash [\Theta],\Gamma}
\]
\[
\infer[\wedgen]{\vdash [\Theta],\Gamma,A\wedgen B}
{\vdash [\Theta],\Gamma,A & \vdash [\Theta],\Gamma,B}
\qquad
\infer[\veen]{\vdash [\Theta],\Gamma, A\veen B}
{\vdash [\Theta],\Gamma, A,B}
\]
\[
\infer[\impln]{\vdash [\Theta],\Gamma, A\impln B}
{\vdash [\Theta], \Gamma, B, \neg A}
\qquad
\infer[\forall]{\vdash [\Theta],\Gamma, \forall x A}
{\vdash [\Theta],\Gamma, A}
\]
\[
\infer[\wedgep]{\mapsto [\Theta], A\wedgep B}
{\mapsto [\Theta],A & \mapsto [\Theta],B }
\qquad
\infer[\veep]{\mapsto [\Theta],A_1\veep A_2}
{\mapsto [\Theta], A_i}
\qquad
\infer[\exists] {\mapsto [\Theta], \exists x A}
{\mapsto [\Theta], A[t/x]}
\]
\[
\infer[\implp]{\mapsto [\Theta], A\implp B}
{\mapsto [\Theta], \neg A}
\qquad
\infer[\implp]{\mapsto [\Theta], A\implp B}
{\mapsto [\Theta], B}
\]
\caption{The Classical Sequent Calculus LKF.  Here, $P$ is positive,
$N$ is negative, $C$ is a positive formula or a negative literal,
$\Theta$ consists of positive formulas and negative literals, and $x$
is not free in $\Theta$, $\Gamma$.  End-sequents have the form $\vdash
[],\Gamma$.}
\label{lkf}
\end{figure}

The following correctness theorem for LKF can be proved by 
relating it to the G\"odel-Gentzen translation (see
\cite[Section 9]{liang06report} for more details).
\begin{theorem}
LKF is sound and complete with respect to classical logic.
\end{theorem}

We have constructed this embedding of classical logic as a further
demonstration of the abilities of \ljf\ as a hosting framework.  The
embedding also revealed interesting relationships between classical
and intuitionistic polarity.  It is also possible to
derive LKF from linear logic: 
each connective needs to be defined as either wholly positive or negative.  
For example, the translation of $(A\veen B)^p$ is
$A^p\lpar B^p$ if $A^p$ and $B^p$ are both negative; is
$A^p\lpar \quest B^p$ if only $A^p$ is negative; is
$\quest A^p\lpar B^p$  if only $B^p$ is negative; and is
$?A^p\lpar ?B^p$,  if $A^p$ and $B^p$ are both positive.
This translation is called the {\em ``polaro''\/}
translation in \cite{danos97jsl}, where it was used to formulate
$LK_{p}^\eta$, the first focused proof
system for classical logic.  Like the $\approx$ translation, the
polaro translation is a derivative of the LC/LU analysis of polarity.
Except for the treatment of atoms,
LKF is derivable from LLF using the polaro translation
in the same manner that \ljf\ is derived.

$LK_{p}^\eta$ was extended to $LK_{pol}^{\eta,\rho}$ in
\cite{laurent05apal}.  These systems were formulated independently of
Andreoli's results.  The authors of \cite{danos97jsl} opted not to
present $LK_p^\eta$ as a sequent calculus because they feared that it
will have the cumbersome size of LU.  Such cumbersomeness can, in
fact, be avoided by adopting LLF-style {\em reaction\/} rules.

Given our goals, the choice in adopting Andreoli's system is justified
in that LKF and LJF have the form of compact sequent calculi ready for
implementation.  More significantly perhaps, $LK_{p}^\eta$ and
$LK_{pol}^{\eta,\rho}$ define focusing for classical logic.  They map
to polarized forms of linear logic (LLP and $\mbox{LL}_{pol}$).  LLF
is defined for full classical linear logic.  LKF is embedded within
LLF in the same way that LC is embedded within LU.  LLF is well suited
for hosting other logics.

\section{Conclusion and Future Work}

We have studied focused proof construction in intuitionistic logic.
The key to this endeavor is the definition of polarity for
intuitionistic logic. The \ljf\ proof system captures focusing
using this notion of polarity.  We illustrate how systems such as LJ,
LJT, LJQ, and $\lambda$RCC can be captured within \ljf\ by assigning
polarity to atoms and by 
adding to intuitionistic logic formulas annotations on conjunctions
and delaying operators.  We also use \ljf\ to derive and
justify the \lkf\ focusing proof system for classical logic.

It remains to examine the impact of these focusing calculi on typed
$\lambda$-calculi, logic programming, and theorem proving.  Given the
connections observed between LJQ/LJT and call-by-name/value, the \ljf\
system could provide a framework for $\lambda$-term evaluations that
combine the eager and lazy evaluation strategies.  In the area of
theorem proving, there are a number of completeness theorems for
various restrictions to resolution: it would be interesting to see
if any of these are captured by an appropriate mapping into \lkf.

\paragraph{Acknowledgments}
We would like to thank the reviewers of an earlier version of this paper
for their comments.
The work reported here was carried out while the first author was on
sabbatical leave from Hofstra University to LIX.
The second author was supported by INRIA through the ``Equipes
Associ{\'e}es'' Slimmer and by the Information Society Technologies
program of the European Commission, Future and Emerging Technologies
under the IST-2005-015905 MOBIUS project.

\end{document}